\documentclass[%
reprint,
superscriptaddress,
amsmath,amssymb,
aps,
pra,
]{revtex4-2}

\usepackage{graphicx}

\usepackage{dcolumn}
\usepackage{bm}
\usepackage[export]{adjustbox}
\usepackage{soul}
\usepackage{cancel}

\usepackage{mathtools} 
\DeclarePairedDelimiter\bra{\langle}{\rvert}
\DeclarePairedDelimiter\ket{\lvert}{\rangle}
\DeclarePairedDelimiterX\braket[2]{\langle}{\rangle}{#1 \delimsize\vert #2}

\usepackage[normalem]{ulem}

\usepackage[utf8]{inputenc}
\usepackage[T1]{fontenc}

\usepackage{mathrsfs}
\usepackage{dsfont}

\usepackage{amssymb}
\usepackage{amsmath}
\usepackage{amsfonts}
\usepackage{amsthm}

\usepackage{stackengine}
\stackMath

\usepackage{hyperref}
\hypersetup{pdfpagemode=UseNone}

\newtheorem{theorem}{Theorem}

\newtheorem{proposition}[theorem]{Proposition}

\newcounter{rem}
\def\tr{{\rm tr}}

\def\rho{{\varrho}}

\def\textbf#1{{\bf #1}}

\gdef\lvert{\delimiter"426A30C }
\gdef\rvert{\delimiter"526A30C }

\usepackage{xcolor}

\begin{document}

\title{Equilibration of generalized subsystems: a quantum-channel approach}

\author{Pedro S. Correia}
\affiliation{Departamento de Ciências Exatas, Universidade Estadual de Santa Cruz, Ilhéus, Bahia 45662-900, Brazil}
\email{pscorreia@uesc.br}
\author{Adalberto D. Varizi}
\affiliation{Instituto de Física, Universidade Federal Fluminense, Avenida Litoranea s/n, Gragoatá 24210-346, Niterói, Rio de Janeiro, Brazil}
\author{Gabriel Dias Carvalho}
\affiliation{Física de Materiais, Universidade de Pernambuco, 50720-001, Recife, Pernambuco, Brazil}

\date{\today}

\begin{abstract}
Quantum systems governed by unitary and reversible microscopic dynamics may nevertheless exhibit equilibration, in the sense that some effective description becomes time independent. Standard equilibration results usually consider two separate situations: system-environment structures, in which the composite system evolves unitarily while the system of interest equilibrates, and restricted measurements, such as coarse-grained POVMs and observables, in which the measurement statistics equilibrate. Here, we bring these descriptions into a common state-level framework using the concept of generalized subsystems, where the accessible effective state appears as the output of a quantum channel acting on the microscopic state. We derive bounds showing that generalized subsystems equilibrate when their dimension is small compared with the effective dimension of the discarded microscopic information.
We further show that this condition is met for typical initial states in large subspaces and that the resulting equilibrium description is largely insensitive to microscopic initial details. The framework recovers the usual equilibration bounds for ordinary subsystems and finite families of POVMs. As an example, we also introduce a finite-resolution energy channel that maps unresolved microscopic energy levels into effective energy levels, thereby making residual effective coherences explicit and showing how spectral multiplicities constrain those coherences while strengthening equilibration. Our results provide a unified state-level formulation of quantum equilibration under general forms of limited accessible information.
\end{abstract}

\maketitle

\section{Introduction}

For isolated systems, equilibration is the emergence of effectively time-independent statistics for accessible degrees of freedom in large systems, despite microscopically reversible dynamics. In classical statistical mechanics, a standard explanation relies on coarse-graining and typicality: many microstates give rise to the same macroscopic description. A gas at equilibrium, for instance, appears stationary at the macroscopic level because many different microscopic configurations of positions and momenta correspond to the same values of macroscopic variables such as volume, pressure, and temperature. The underlying microscopic state nevertheless continues to evolve. 

A similar picture holds for equilibration in closed
quantum systems: although the microscopic state continues to
evolve unitarily, effective time-independence
emerges at the level of a description that retains only limited
accessible information about that state. Notably, this does not require the resulting stationary description to coincide with a particular thermodynamic
ensemble, such as a microcanonical ensemble or a canonical
Gibbs state. Here, we focus exclusively on equilibration
in this dynamical sense.

A paradigmatic setting for quantum equilibration is the subsystem picture, in which a large isolated many-body system is partitioned into a system of interest and an environment: a small subset of accessible particles defines the subsystem $S$, while the remaining untracked particles form the environment $E$, as illustrated in Fig.~\ref{fig:restricted-access-examples}(a). Although the global state evolves unitarily, the accessible information is encoded in the reduced state of the subsystem. For Hamiltonians satisfying suitable non-resonance conditions, a key result \cite{linden2009, gogolin2016} shows that this reduced state equilibrates, i.e., remains, for almost all times, close to a stationary state. In particular, equilibration becomes more pronounced when the environment is sufficiently large, as information initially accessible in $S$ is redistributed into inaccessible environmental degrees of freedom and system-environment correlations.

\begin{figure*}[t!]
\centering

\begin{minipage}[c]{0.24\textwidth}
    \centering
    \includegraphics[width=3.0cm]{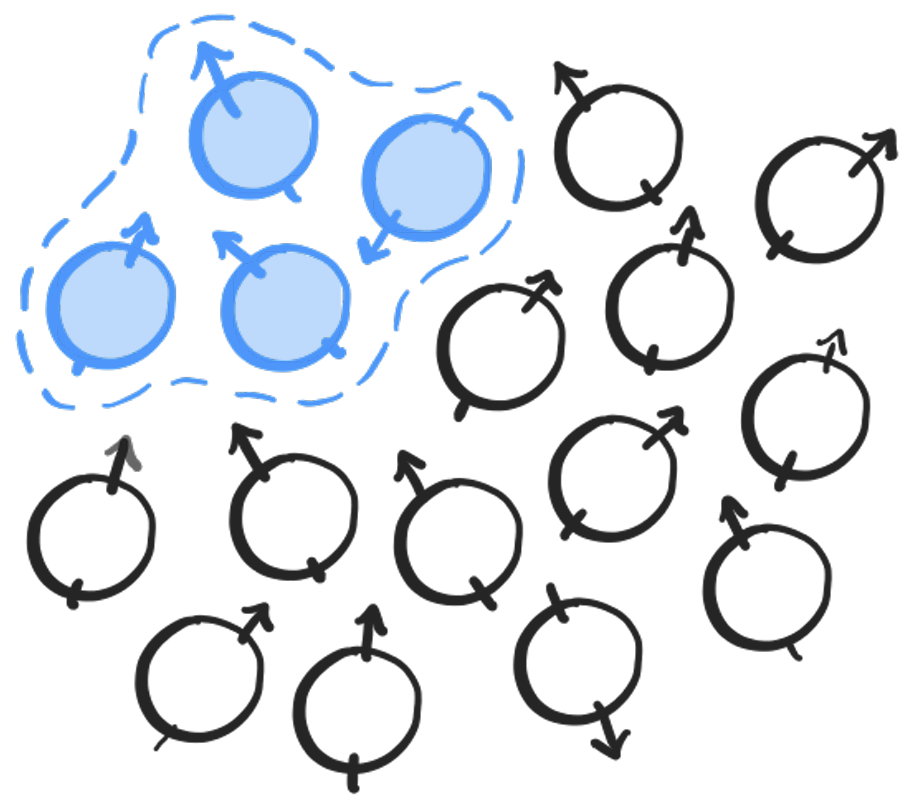}
\end{minipage}
\hfill
\begin{minipage}[c]{0.48\textwidth}
    \centering
    \includegraphics[width=6.5 cm]{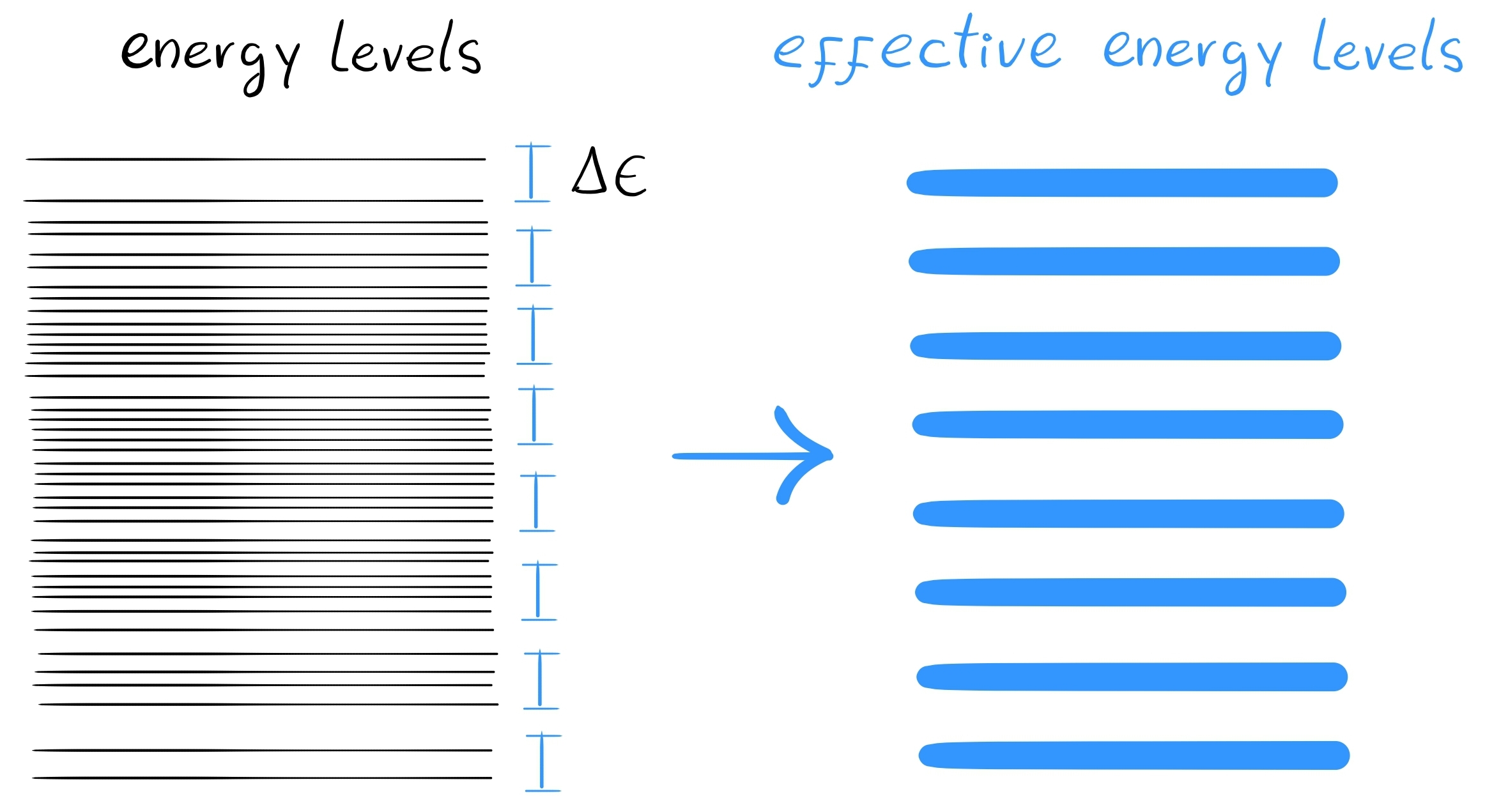}
\end{minipage}
\hfill
\begin{minipage}[c]{0.24\textwidth}
    \centering
    \includegraphics[width=2.8cm]{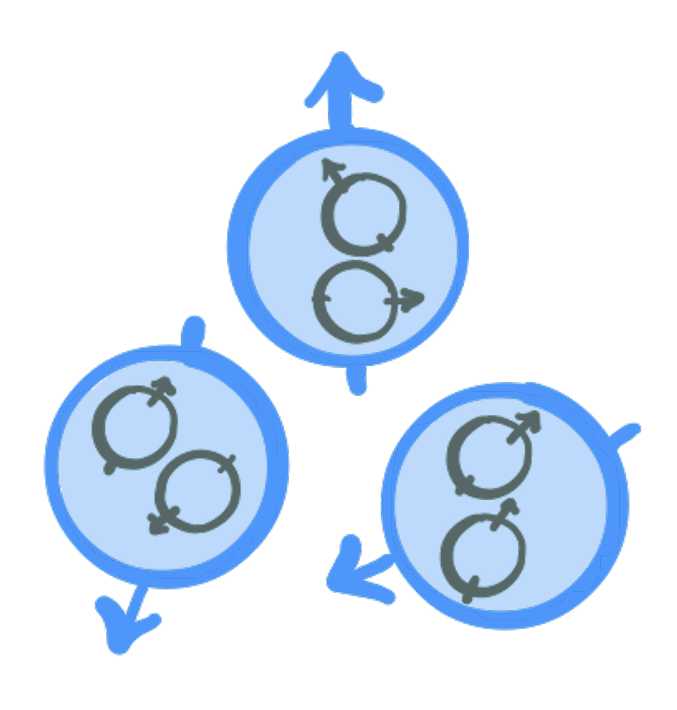}
\end{minipage}

\vspace{0.15cm}

\begin{minipage}[c]{0.24\textwidth}
    \centering
    \textbf{(a)}
\end{minipage}
\hfill
\begin{minipage}[c]{0.48\textwidth}
    \centering
    \textbf{(b)}
\end{minipage}
\hfill
\begin{minipage}[c]{0.24\textwidth}
    \centering
    \textbf{(c)}
\end{minipage}

\caption{\textbf{Examples of limited accessible information.}
\textbf{(a)} In the standard subsystem picture, only a subset of particles is accessible, while the rest forms an inaccessible environment.
\textbf{(b)} Finite-resolution energy measurements group microscopic energy levels into effective energy windows, so that only a coarse-grained energy description is accessible.
\textbf{(c)} In a blurred and saturated detector, distinct microscopic configurations can become indistinguishable at the accessible level; for instance, two microscopic qubits may be described by a single effective two-level system.}
\label{fig:restricted-access-examples}
\end{figure*}

Measurement-based equilibration provides another setting in which equilibration emerges from limited accessible information, now determined not by a subsystem partition but by the information obtainable from a restricted set of feasible measurements performed on the system. Finite-resolution energy measurements provide a physically important example of this situation. In macroscopic systems, finite energy resolution can encompass an overwhelmingly large spectral multiplicity \cite{reimann2008,reimann2010,reimann2012,reimann2013,passos2025}, so that measurements cannot resolve individual microscopic energy values but only coarse-grained energy windows, as illustrated in Fig.~\ref{fig:restricted-access-examples}(b). Such observations can be described by coarse-grained observables or POVMs, and equilibration means that the corresponding accessible information becomes, for almost all times, indistinguishable from that obtained from the stationary time-averaged state. Results show that this behavior is favored when the global state effectively occupies many energy levels compared with the information resolved by the accessible observations \cite{reimann2008,short2011,short2012}.

Although subsystem-based and measurement-based equilibration are formulated in different terms, they express the same observer-dependent situation: only part of the information contained in the evolving microscopic state is operationally available. In the subsystem setting, this information is encoded in a reduced state on a tensor factor; in the measurement setting, it is encoded in expectation values or outcome statistics for a feasible set of measurements. This shared viewpoint, centered on limited accessible information, motivates a unified treatment.

The present work brings these different formulations into a unified treatment by employing the notion of \emph{generalized subsystems}, introduced in Ref.~\cite{alicki2009} and further developed in Refs.~\cite{duarte2017,kabernik2018,silva2019,carvalho2020,kabernik2020,pineda2021,vallejos2022,correia2024,castillo2025}. In this framework, formally introduced in Sec.~\ref{sec:gen subs}, restricted access to a microscopic global system is represented by a quantum channel, namely a completely positive, trace-preserving (CPTP) map. This map assigns to each microscopic state an effective accessible state, which we interpret as the state of a generalized subsystem. Thus, the channel connects two levels of description: the microscopic state of the full system and a state-level description available to the observer. In this way, ordinary subsystems, POVMs, and coarse-grained descriptions can be treated within the same language. As an illustration, the blurred-and-saturated detector introduced in Refs.~\cite{duarte2017,silva2019} and shown in Fig.~\ref{fig:restricted-access-examples}(c) maps two microscopic qubits into a single two-level system, providing an effective description that is not reducible to an ordinary partial trace.

Equilibration then becomes a natural question for generalized subsystems: under what conditions does the state of a generalized subsystem equilibrate under the microscopic unitary dynamics? The main result of this work, presented in Sec.~\ref{sec:gen equi}, answers this question by providing a bound on the time-averaged deviation between the state of the generalized subsystem and its stationary time-averaged state. The bound shows that equilibration is controlled by the size of the generalized subsystem relative to the effective dimension explored by the microscopic dynamics. This intuition is no longer tied to a fixed subsystem--environment partition, but applies whenever the effective description can be represented by a generalized subsystem. We also show that the condition for equilibration is typically satisfied for random initial states in large subspaces, and that the resulting equilibrium description becomes largely insensitive to microscopic initial details. A previous work~\cite{correia2024} had shown that generalized subsystems are subject to canonical typicality.
That is, almost all compatible microscopic pure states yield, on a sufficiently small generalized subsystem, an effective state close to the corresponding reduced ensemble average.
In this respect, the results here provide a dynamical counterpart to those in~\cite{correia2024}.

To illustrate the scope and physical significance of the generalized-subsystem framework, we develop two applications in Sec.~\ref{sec:applications}. First, we show that standard equilibration bounds for finite families of POVMs are recovered as state-level bounds for the generalized subsystems associated with those measurements. In this way, equilibration described through measurements is placed in the same effective-state language as equilibration described through subsystems. Second, we introduce a finite-resolution energy channel that maps unresolved microscopic energy levels into effective energy levels. Unlike standard projective coarse-graining, which retains only coarse energy populations, this channel produces an effective quantum state and therefore makes any residual coherences between effective energy levels explicit. This allows us to show how finite spectral resolution both constrains the coherences that remain operationally accessible and strengthens equilibration through the large multiplicities of unresolved microlevels.

Taken together, the recovery of the usual subsystem result of Ref.~\cite{linden2009} as the partial-trace case, the POVM formulation, and the finite-resolution energy channel show that different formulations of equilibration can be understood as the equilibration of the effective state of a corresponding generalized subsystem. The framework therefore provides a common state-level language for situations in which only limited information about a unitarily evolving microscopic state is operationally accessible, whether this limitation is described by a tensor-product partition, a restricted family of measurements, finite spectral resolution, or any other generalized-subsystem structure.

\section{Generalized Subsystems}
\label{sec:gen subs}

Consider a microscopic global quantum system that we want to analyze only through some effective level of description. Let \(\mathcal{H}_U\) denote its Hilbert space, and let \(\rho\in\mathcal{L}(\mathcal{H}_U)\) be its state. As discussed in the Introduction, the usual subsystem picture provides a standard example of such an effective description. In this setting, the global Hilbert space is assumed to factorize as
\(\mathcal{H}_U=\mathcal{H}_S\otimes\mathcal{H}_E\), where \(\mathcal{H}_S\) is the subsystem of interest and \(\mathcal{H}_E\) is the environment. The accessible description is given by the reduced state of the subsystem, 
\(\rho_S=\tr_E[\rho]\), with
\(\tr_E:\mathcal{L}(\mathcal{H}_S\otimes\mathcal{H}_E)\to\mathcal{L}(\mathcal{H}_S)\).
Thus, in the usual subsystem picture, the effective information is fixed by a tensor-product partition and obtained through a partial trace.

As discussed above, more general descriptions may arise from other operational restrictions that are not reducible to the usual subsystem framework, such as coarse-grainings, finite resolutions, and related forms of limited access. In the blurred-and-saturated detector~\cite{duarte2017,silva2019}, for example, two neighboring qubits are probed through emitted light. The detector is blurred, so it cannot identify which qubit emitted, and saturated, so it cannot distinguish one emitted photon from two. Hence, the states \(\ket{01}\), \(\ket{10}\), and \(\ket{11}\) become operationally indistinguishable and are represented by a single effective excited state \(\ket{1_\Lambda}\), while \(\ket{00}\) is mapped to an effective ground state \(\ket{0_\Lambda}\). Such a mapping cannot be obtained from a partial trace, even allowing for a unitary rearrangement of the microscopic degrees of freedom \footnote{If it could be written as \(\rho_S=\tr_E[U\rho U^\dagger]\), the three orthogonal states \(\ket{01}\), \(\ket{10}\), and \(\ket{11}\), all mapped to \(\ket{1_\Lambda}\), would require three orthogonal environmental states. Thus \(\dim\mathcal H_E\ge 3\), so \(\dim(\mathcal H_S\otimes\mathcal H_E)\ge 6\), contradicting the original dimension \(4\).}.

A second elementary example is provided by finite energy resolution in a three-level system, which can be viewed as a simple qutrit-to-qubit coarse-graining related to effective descriptions studied in other contexts~\cite{duarte2020}. Consider a microscopic system with three energy eigenstates \(\ket{E_0}\), \(\ket{E_1}\), and \(\ket{E_2}\). Suppose that the ground level \(\ket{E_0}\) is resolved, while the two excited levels \(\ket{E_1}\) and \(\ket{E_2}\) are separated by an energy difference smaller than the resolution of the apparatus. Operationally, the measurement does not distinguish these two excited levels, but only whether the system is in the ground level or in the unresolved excited energy window. The effective description is therefore two-dimensional: \(\ket{E_0}\) is represented by an effective ground state \(\ket{0_\Lambda}\), whereas both \(\ket{E_1}\) and \(\ket{E_2}\) are represented by the same effective excited state \(\ket{1_\Lambda}\). Again, this effective qubit cannot be interpreted as an ordinary subsystem obtained by tracing out an environment: a three-dimensional microscopic Hilbert space cannot be decomposed into a two-dimensional subsystem and an additional tensor factor.

The effective descriptions in these examples motivate the broader notion of a \emph{generalized subsystem}. In this framework, the state of a generalized subsystem is described as the output of a quantum channel \(\Lambda\)~\cite{duarte2017,kabernik2018,silva2019,carvalho2020,kabernik2020,pineda2021,vallejos2022,correia2024,castillo2025,duarte2020,correia2021},
\begin{equation}
    \rho_S=\Lambda[\rho],
    \qquad
    \Lambda:\mathcal{L}(\mathcal{H}_U)\to\mathcal{L}(\mathcal{H}_S),
    \label{eq:gsub_channel}
\end{equation}
with \(d_S=\dim(\mathcal{H}_S)\le d_U=\dim(\mathcal{H}_U)\). In this way, the channel \(\Lambda\) connects the microscopic state \(\rho\) to the effective state \(\rho_S\) assigned to the generalized subsystem.

\section{Equilibration of Generalized Subsystems}
\label{sec:gen equi}

With the generalized-subsystem framework established, we now study the conditions under which the state of a generalized subsystem equilibrates under microscopic unitary dynamics, extending familiar subsystem equilibration bounds beyond the standard tensor-factor setting.

\subsection{The microscopic system} 

Consider a large isolated quantum system associated with a Hilbert space $\mathcal{H}_U$ of dimension $d_U$. The global dynamics is governed by a Hamiltonian \(H = \sum_j E_j \ket{E_j}\bra{E_j}\), where \(\ket{E_j} \in \mathcal{H}_U\) are energy eigenstates with eigenvalues \(E_j\). For simplicity, we assume that \(H\) has non-degenerate energy gaps, but the following results can be extended to accommodate these scenarios~\cite{short2012,Reimann2012a}.
As results are discussed, we highlight the differences imposed by these conditions.

Moreover, we consider the initial microscopic state to be pure, however the final equilibration result in Sec.~\ref{subsec:GenEqui} also holds for initial mixed states, albeit in a weaker form.
At time \(t\), the global system is described by \(\ket{\psi(t)} \in \mathcal{H}_U\), which can be expressed in the energy eigenbasis, setting \(\hbar=1\), as \(\ket{\psi(t)} = \sum_j c_j e^{-iE_jt}\ket{E_j}\), with \(\sum_j |c_j|^2 = 1\). For convenience, the pure state can also be represented by the corresponding density operator \(\psi(t)=\ket{\psi(t)}\bra{\psi(t)}\), which reads
\begin{align}
    \psi(t)
    = \sum_j |c_j|^2 \ket{E_j}\bra{E_j}
    + \sum_{j\neq k} c_j c_k^\ast e^{-i(E_j-E_k)t}
    \ket{E_j}\bra{E_k}.
    \label{eq:globalpure}
\end{align}

The infinite-time averaged state is \(\omega=\langle\psi(t)\rangle_t =\lim_{\tau\rightarrow\infty}\frac{1}{\tau}\int_0^\tau \psi(t)\,dt \). This time average removes the oscillatory coherences between distinct energy eigenspaces, leaving only the diagonal energy populations:
\begin{align}
    \omega &= \sum_j |c_j|^2 \ket{E_j}\bra{E_j}.
    \label{eq:omega}
\end{align}

Thus, although \(\psi(t)\) remains pure throughout the unitary evolution, the time-averaged state \(\omega\) is stationary and retains the conserved energy populations \(|c_j|^2\). In this sense, \(\omega\) encodes the distribution over energy eigenstates that underlies the microscopic evolution.

To quantify how many energy eigenstates effectively participate in this distribution, we use the effective dimension. Formally, for any state \(\rho\), we define it as
\begin{align}
    d_{\mathrm{eff}}(\rho):=\frac{1}{\tr(\rho^2)} .
    \label{eq:deff}
\end{align}
Unlike the rank, which only counts the number of nonzero eigenvalues, \(d_{\mathrm{eff}}(\rho)\) also depends on how evenly the eigenvalues are distributed. It is large when many eigenvectors contribute with comparable weights, and small when the state is concentrated on only a few eigenvectors. For the time-averaged state \(\omega\), this general definition becomes
\[
    d_{\mathrm{eff}}(\omega)=\frac{1}{\sum_j |c_j|^4},
\]
which measures the effective number of energy eigenstates involved in the evolution of \(\psi(t)\).

\subsection{\label{subsec:GenEqui}Generalized Subsystem Equilibration}

Having specified the underlying unitary dynamics and its stationary time-averaged state, we now ask how this dynamics appears through the effective description defined by a generalized subsystem. Given the global pure state \(\psi(t)\) of Eq.~\eqref{eq:globalpure}, the state of the generalized subsystem at time \(t\) is
\begin{align}
    \rho_S(t) 
    = \Lambda[\psi(t)] 
    = \sum_{j,k} c_j c_k^\ast e^{-i(E_j-E_k)t}\,\Lambda(\ket{E_j}\bra{E_k}),
    \label{eq:gsub}
\end{align}
where \(\Lambda: \mathcal{L}(\mathcal{H}_U)\to \mathcal{L}(\mathcal{H}_S)\).  
The time-averaged state of the generalized subsystem is
\begin{align}
    \omega_S 
    = \Lambda[\omega] 
    = \sum_j |c_j|^2 \Lambda(\ket{E_j}\bra{E_j}),
    \label{eq:gsubeq}
\end{align}
where \(\omega\) is the global time-averaged state \eqref{eq:omega}.

We say that the generalized subsystem equilibrates when \(\rho_S(t)\) remains, for most times, indistinguishable from its time-averaged state \(\omega_S\), even though the microscopic state \(\psi(t)\) continues to evolve unitarily. In this sense, equilibrium is not a property of the full microscopic state, but of the effective state of the generalized subsystem defined by \(\Lambda\). This motivates quantifying equilibration by the time-averaged trace distance between \(\rho_S(t)\) and \(\omega_S\), for which we obtain the following bound (derived in App.~\ref{ap:mainresults}):
\begin{align}
\label{eq:main-result}
    \big\langle D\bigl(\rho_S(t), \omega_S\bigr) \big\rangle_t 
    \le \frac{1}{2}\sqrt{\frac{d_S}{d_{\mathrm{eff}}(\Lambda_c[\omega])}} .
\end{align}
Here \(D(\rho,\sigma)=\tfrac{1}{2}\|\rho-\sigma\|_1\) is the trace distance, and \(d_{\mathrm{eff}}\) is the effective dimension defined in Eq.~\eqref{eq:deff}. The map \(\Lambda_c\) is a complementary channel associated with \(\Lambda\): for a Stinespring dilation \(V:\mathcal H_U\to\mathcal H_S\otimes\mathcal H_{\mathrm{aux}}\), such that \(\Lambda[\rho]=\tr_{\mathrm{aux}}[V\rho V^\dagger]\), it is defined by \(\Lambda_c[\rho]=\tr_S[V\rho V^\dagger]\). Thus, \(\Lambda\) gives the accessible effective state, whereas \(\Lambda_c\) gives the complementary output associated with the same dilation. For a general channel, \(\Lambda_c[\rho]\) should not be interpreted as a physical environment, but as a complementary Stinespring output that records microscopic information not resolved in the generalized-subsystem state \(\Lambda[\rho]\).

The complementary channel is not unique, because the Stinespring dilation is not unique. However, this does not make the bound ambiguous: \(d_{\mathrm{eff}}(\Lambda_c[\omega])\) is independent of the chosen dilation and is therefore an intrinsic quantity of the pair \((\Lambda,\omega)\).\footnote{One may choose a minimal Stinespring dilation \(V_{\min}:\mathcal H_U\to\mathcal H_S\otimes\mathcal H_{\mathrm{aux},\min}\), with \(\dim\mathcal H_{\mathrm{aux},\min}=\operatorname{rank}(J_\Lambda)\), where \(J_\Lambda=(\Lambda\otimes\mathrm{id})(\ket{\Omega}\bra{\Omega})\), \(\ket{\Omega}=\sum_i\ket{i}\otimes\ket{i}/\sqrt{d_U}\). Any other dilation has the form \(V=(\mathbb I_S\otimes W)V_{\min}\), with \(W\) an isometry, so \(\Lambda_c[\omega]=W\Lambda_{c,\min}[\omega]W^\dagger\). Hence the nonzero spectrum, and therefore \(d_{\mathrm{eff}}\), is unchanged.} The two ingredients of the problem therefore enter together: the microscopic dynamics and the initial state determine \(\omega\), while \(\Lambda\) determines the accessible description.

The bound extends the usual subsystem equilibration mechanism to arbitrary generalized subsystems. For the usual subsystem case, \(\mathcal H_U=\mathcal H_S\otimes\mathcal H_E\) and \(\Lambda=\tr_E\). This corresponds to the Stinespring isometry \(V=\mathbb I_{\mathcal H_S\otimes\mathcal H_E}\), with \(\mathcal H_{\mathrm{aux}}\equiv\mathcal H_E\), so that the complementary channel is \(\Lambda_c=\tr_S\) and \(\Lambda_c[\omega]=\tr_S[\omega]\), the time-averaged state of the environment. In this special case, the complementary output is literally the state of the inaccessible environment, and \(d_{\mathrm{eff}}(\tr_S[\omega])\) quantifies the effective number of eigenstates contributing to the environmental state, through the inverse of its purity. Equation~\eqref{eq:main-result} then reduces to the equilibration bound of Ref.~\cite{linden2009}. For an arbitrary generalized subsystem, there need not be any microscopic system--environment factorization. Nevertheless, the same structural idea survives at the level of channels: the Stinespring dilation associates an effective environment with the generalized subsystem, and \(\Lambda_c[\omega]\) is the corresponding complementary state. Equilibration is then favored whenever this effective environment has large effective dimension compared with \(d_S\).

A looser consequence of Eq.~\eqref{eq:main-result}, also derived in App.~\ref{ap:mainresults}, involves only \(d_S\) and the effective dimension of the global time-averaged state:
\begin{align}
\label{eq:main-result-envfree}
    \big\langle D\bigl(\rho_S(t), \omega_S\bigr) \big\rangle_t 
    \le \frac{d_S}{2\sqrt {d_{\mathrm{eff}}(\omega)}} .
\end{align}
This second inequality does not require evaluating the complementary channel. It follows from
\(\tr[(\Lambda_c[\omega])^2]\le d_S\,\tr(\omega^2)\), which implies
\(d_{\mathrm{eff}}(\Lambda_c[\omega])\ge d_{\mathrm{eff}}(\omega)/d_S\). Thus, while Eq.~\eqref{eq:main-result} gives the sharper bound in terms of the complementary output associated with the pair \((\Lambda,\omega)\), Eq.~\eqref{eq:main-result-envfree} provides a simpler sufficient criterion for equilibration. It depends only on the dimension \(d_S\) of the generalized subsystem and on the effective dimension explored by the microscopic dynamics, encoded in \(\omega\). In particular, whenever \(d_{\mathrm{eff}}(\omega)\gg d_S^2\), the generalized subsystem equilibrates.

The results in Eqs.~\eqref{eq:main-result} and~\eqref{eq:main-result-envfree} are readily extended to degenerate Hamiltonians~\cite{short2011,reimann2010}.
That is, for $H=\sum_j E_j P_j$, where $P_j$ is the projector onto the energy eigenspace with eigenvalue $E_j$, the aforementioned bounds remain the same with the time-averaged state and effective dimension given by \(\omega= \sum_j P_j |\psi(0)\rangle \langle\psi(0)| P_j\) and \(d_\mathrm{eff}(\omega) =  1/\sum_j [\tr(|\psi(0)\rangle\langle\psi(0)| P_j)]^2\).
Moreover, using the results on equilibration of observables in~\cite{short2011} and their purification approach, Eq.~\eqref{eq:main-result-envfree} can be easily extended to mixed states. 

The presence of degeneracies in the energy gaps worsen the bounds on the equilibration of any observable, and consequently of effective states, by a factor related to the maximum degeneracy of energy gaps~\cite{short2012,Reimann2012a}. However, since $d_\mathrm{eff}$ is usually exponentially large in the system size, systems with a non-exponentially large number of degenerate energy gaps will typically exhibit equilibration in the sense considered here.

\subsection{Typicality in Generalized-Subsystem Equilibration}

Equilibration in generalized subsystems is typical in two related senses. 

First, for most initial states in a large subspace, the microscopic time-averaged state has a large effective dimension \(d_{\mathrm{eff}}(\omega)\). So Eq.~\eqref{eq:main-result-envfree} guarantees equilibration of any generalized subsystem of sufficiently small dimension. 
Second, for a fixed channel \(\Lambda\) defining the generalized subsystem, the corresponding equilibrium effective state depends only weakly on the microscopic initial state, yielding approximately the same description for most initial microscopic pure states.

To make these statements precise, let the initial microscopic state \(\ket{\psi}\) be drawn uniformly, i.e., according to the Haar measure, from a subspace \(\mathcal H_R\subseteq\mathcal H_U\) of dimension \(d_R\). In standard statistical mechanics, for example, \(\mathcal{ H}_R\) may be a subspace selected by a fixed energy shell. Let \(\overline{\vphantom{X}(\cdot)}^{\,\psi}\) denote the corresponding ensemble average. It is known that \(d_{\mathrm{eff}}(\omega)\) strongly concentrates at large values, with
\(\overline{d_{\mathrm{eff}}(\omega)}^{\,\psi} \ge d_R/2\) and
\(\Pr\{d_{\mathrm{eff}}(\omega)<d_R/4\}\) exponentially small in \(\sqrt{d_R}\)~\cite{linden2009}. Inserting this typical behavior into Eq.~\eqref{eq:main-result-envfree} yields, with exponentially high probability,
\begin{align}
\big\langle D\big(\rho_S(t),\omega_S\big)\big\rangle_t
\;\le\; \frac{d_S}{\sqrt{d_R}}\,.
\end{align}
Therefore, for typical random states in sufficiently large subspaces---those for which $d_\mathrm{eff}(\omega)\ge d_R/4$--- generalized subsystems equilibrate, with deviations suppressed as \(\sim d_S/\sqrt{d_R}\).
Thus, the condition required by Eq.~\eqref{eq:main-result-envfree} is naturally satisfied for any generalized subsystem with \(d_S\ll \sqrt{d_R}\).

The second sense of typicality concerns the dependence of the equilibrium state on the microscopic initial state. For each uniformly sampled realization \(\ket{\psi}\in\mathcal{H}_R\subseteq\mathcal{H}_U\), let \(\omega^\psi\) be the associated microscopic time-averaged state and
\(\omega_S^\psi:=\Lambda[\omega^\psi]\) the corresponding time-averaged state of the generalized subsystem. Then, as shown in App.~\ref{ap:mainresults}, \(\omega_S^\psi\) concentrates around its ensemble mean \(\overline{\omega_S^\psi}^{\,\psi}\):
\begin{align}
\overline{D\!\left(\omega_S^\psi,\overline{\omega_S^\psi}^{\,\psi}\right)}^{\,\psi}
\;\le\; \frac{1}{2}\sqrt{\frac{d_S}{d_R}}.
\label{eq:concentration-omegaS}
\end{align}
Combined with the equilibration bound~\eqref{eq:main-result-envfree}, this means that, for most typical initial microscopic preparations \(|\psi\rangle\), and for most times, \(\rho_S^\psi(t)\) is close to an equilibrium description that is largely independent of $|\psi\rangle$. More precisely, \(\rho_S^\psi(t)\) equilibrates around its own time-averaged state \(\omega_S^\psi\), while Eq.~\eqref{eq:concentration-omegaS} shows that these states \(\omega_S^\psi\) themselves concentrate around a common ensemble-averaged state $\overline{\omega_S^\psi}^\psi$. 
Thus, for a fixed generalized subsystem, the equilibrium description is only weakly sensitive to the microscopic state chosen, with most pure states within the large subspace yielding nearly identical effective descriptions.

\section{Applications}
\label{sec:applications}

\subsection{POVM equilibration as a generalized subsystem}

As a first application, we consider an observer whose access to the system is restricted to a finite family of POVMs. We show that, by representing each POVM through its associated measurement channel, the generalized-subsystem equilibration bound reproduces the standard POVM equilibration result of Ref.~\cite{short2011}. Thus, equilibration described through measurements is placed in the same effective-state language as ordinary subsystem equilibration. Together with the usual subsystem case recovered in the previous section, this shows that the generalized-subsystem framework contains standard formulations of equilibration as particular cases.

Let \(\mathcal M=\{M^{(m)}\}_{m=1}^{n_{\mathcal M}}\) be a finite family of POVMs available to the observer. Each measurement \(M^{(m)}=\{M^{(m)}_r\}_{r=1}^{R_m}\) has \(R_m\) possible outcomes. Operationally, the observer does not access the full microscopic state \(\psi\), but only the probability distribution generated by each measurement. In the generalized-subsystem framework, this accessible information is represented at the state level by the measurement channel
\begin{equation}
\Lambda_{M^{(m)}}(\psi)=\sum_{r=1}^{R_m}\tr(M^{(m)}_r\psi)\ket r\!\bra r .
\end{equation}
Here \(\{\ket r\}_{r=1}^{R_m}\) spans a Hilbert space \(\mathcal H_{S_m}\) of dimension \(d_{S_m}=R_m\), which represents the classical outcome register of the measurement. Thus, \(\Lambda_{M^{(m)}}(\psi)\) is the effective state of a generalized subsystem, diagonal in the outcome basis, whose entries are the corresponding outcome probabilities \(\tr(M^{(m)}_r\psi)\).

For each measurement channel, Eq.~\eqref{eq:main-result-envfree} bounds the time-averaged trace distance by \(R_m/(2\sqrt{d_{\mathrm{eff}}(\omega)})\). Denoting by \(N(\mathcal M):=\sum_{m=1}^{n_{\mathcal M}}R_m\) the total number of outcomes in the family, we apply the bound above to each POVM separately. Since the trace distances are nonnegative, \(\max_m x_m(t)\le\sum_m x_m(t)\) at each time. Using this inequality and linearity of the time average, we obtain
\begin{align}
\Big\langle \max_{1\le m\le n_{\mathcal M}} D\!\bigl(\Lambda_{M^{(m)}}(\psi(t)),\Lambda_{M^{(m)}}(\omega)\bigr) \Big\rangle_t \le \frac{N(\mathcal M)}{2\sqrt{d_{\mathrm{eff}}(\omega)}}.
\label{eq:povm-family-bound}
\end{align}
Hence, if \(\sqrt{d_{\mathrm{eff}}(\omega)}\) is large compared with \(N(\mathcal M)\), then all measurement-defined generalized subsystems remain close to their equilibrium values for almost all times.

We now connect Eq.~\eqref{eq:povm-family-bound} with the standard statistical formulation of Ref.~\cite{short2011}. Since the measurement-channel outputs are diagonal in the outcome register, their trace distance is simply the total-variation distance between the corresponding outcome distributions. Indeed,
\(D(\Lambda_M(\rho),\Lambda_M(\sigma))=\mathcal D_M(\rho,\sigma)\), where \(\mathcal D_M(\rho,\sigma):= \frac12 \sum_{r=1}^{R(M)} \left|\tr\!\left[M_r(\rho-\sigma)\right]\right|\). Denoting by \(\mathcal D_{\mathcal M}\) the maximal statistical distance over the family \(\mathcal M\), Eq.~\eqref{eq:povm-family-bound} gives \(\langle\mathcal D_{\mathcal M}(\psi(t),\omega)\rangle_t\le N(\mathcal M)/(2\sqrt{d_{\mathrm{eff}}(\omega)})\). Thus, POVM equilibration is recovered here as equilibration of the generalized subsystems associated with the measurement channels, with the same operational form and scaling as in Ref.~\cite{short2011}.

\subsection{Finite spectral resolution as a generalized subsystem}
\label{sec:energy-band}

As highlighted in the Introduction, finite-resolution energy measurements provide a physically important form of limited accessible information. In macroscopic systems, the density of states is so large that individual microscopic energy levels cannot be resolved in practice~\cite{reimann2008,reimann2010,reimann2012,reimann2013,passos2025}. A measurement apparatus with finite resolution, therefore, groups together all energy eigenvalues lying within the same resolution window \(\Delta E\), representing them by a single effective outcome.  This finite-resolution description naturally defines a generalized subsystem whose states are associated with the effective energy outcomes determined by the resolution $\Delta E$.

More specifically, suppose that the resolution \(\Delta E\) defines \(K\) effective energy levels or outcomes. We write the microscopic Hamiltonian of the system by labeling its energy eigenstates according to the effective level to which they belong, as \(H=\sum_{j=0}^{K-1}\sum_{\nu=1}^{d_j} E_{j,\nu}\ket{j,\nu}\bra{j,\nu}\). Here \(j\) labels the effective energy selected by the resolution \(\Delta E\), while \(\nu\) distinguishes the \(d_j\) microscopic energy eigenstates grouped into that effective level. Thus, \(\mathcal H_U\) has total dimension \(d_U=\sum_{j=0}^{K-1}d_j\).

The generalized subsystem associated with this finite-resolution description is represented on a Hilbert space \(\mathcal H_S\), spanned by the states \(\{\ket{j_\Lambda}\}_{j=0}^{K-1}\) associated with the effective energy outcomes. As a consequence, $d_S=\mathrm{dim}(\mathcal{H}_S)=K$. The corresponding channel \(\Lambda_{\Delta E}:\mathcal{L}(\mathcal{H}_U)\to\mathcal{L}(\mathcal{H}_S)\) is defined by
\begin{subequations}
\label{eq:energy-window-map}
\begin{align}
   \Lambda_{\Delta E}\bigl[\ket{j,\mu}\bra{j,\nu}\bigr] 
   &= \delta_{\mu\nu}\,\ket{j_\Lambda}\bra{j_\Lambda}, 
   \label{eq:energy-window-map:diag} \\
   \Lambda_{\Delta E}\bigl[\ket{i,\mu}\bra{j,\nu}\bigr] 
   &= \alpha_{ij}\,\ket{i_\Lambda}\bra{j_\Lambda}, 
   \qquad i\neq j,
   \label{eq:energy-window-map:off}
\end{align}
\end{subequations}
with $\delta_{\mu\nu}$ the Kronecker delta and \(\alpha_{ij}=\alpha_{ji}^\ast\), so that Hermiticity is preserved.

By directly applying this channel to a generic microscopic state \(\rho=\sum_{i,j,\mu,\nu} \rho_{i\mu,\,j\nu}\,\ket{i,\mu}\bra{j,\nu}\), we obtain the effective description
\begin{align}
\label{eq:energy-window-rhoS}
\rho_S 
= \sum_{i=0}^{K-1} 
\left(\sum_{\mu=1}^{d_i}\rho_{i\mu,\,i\mu}\right)&
\ket{i_\Lambda}\bra{i_\Lambda}
+ \nonumber \\
&+\sum_{\substack{i,j=0\\ i\neq j}}^{K-1}
\alpha_{ij}
\sum_{\mu=1}^{d_i}\sum_{\nu=1}^{d_j}
\rho_{i\mu,\,j\nu}\,
\ket{i_\Lambda}\bra{j_\Lambda}.
\end{align}

An important feature of this formulation is that the coefficients \(\alpha_{ij}\) are not freely prescribed. As shown in Appendix~\ref{ap:finite-resolution-cptp}, complete positivity of \(\Lambda_{\Delta E}\) implies, for every pair \(i\neq j\), the constraint \(|\alpha_{ij}|\le 1/\sqrt{d_i d_j}\). Thus, coherences between distinct energy windows need not vanish a priori, but their magnitude is bounded by the window sizes \(d_i\) and \(d_j\).

This state-level description therefore contains more information than the effective energy-level populations alone: it also tracks the coherences that may remain visible between distinct effective energy levels. However, these coherences are not arbitrary. Requiring the map to be CPTP imposes bounds on the coefficients \(\alpha_{ij}\), showing that their contribution is progressively suppressed as the coarse-graining becomes stronger. For smaller or intermediate windows, residual coherences may still be present. In the macroscopic regime, where \(d_i,d_j\gg 1\), the constraint forces \(\alpha_{ij}\) to be extremely small, so that the effective state becomes practically diagonal in the finite-resolution energy basis. In this regime, our description recovers the same physical picture as the usual projective coarse-grained description, in which energy windows are represented by projectors onto microscopic subspaces~\cite{vsafranek2019, buscemi2023, rubino2025coarse}. The difference is that, here, the disappearance of the residual effective coherences follows from the CPTP constraint together with large spectral multiplicities, rather than being imposed a priori.

\begin{figure*}[t!]
    \centering
    \begin{minipage}{0.49\textwidth}
        \centering
        \includegraphics[width=\linewidth]{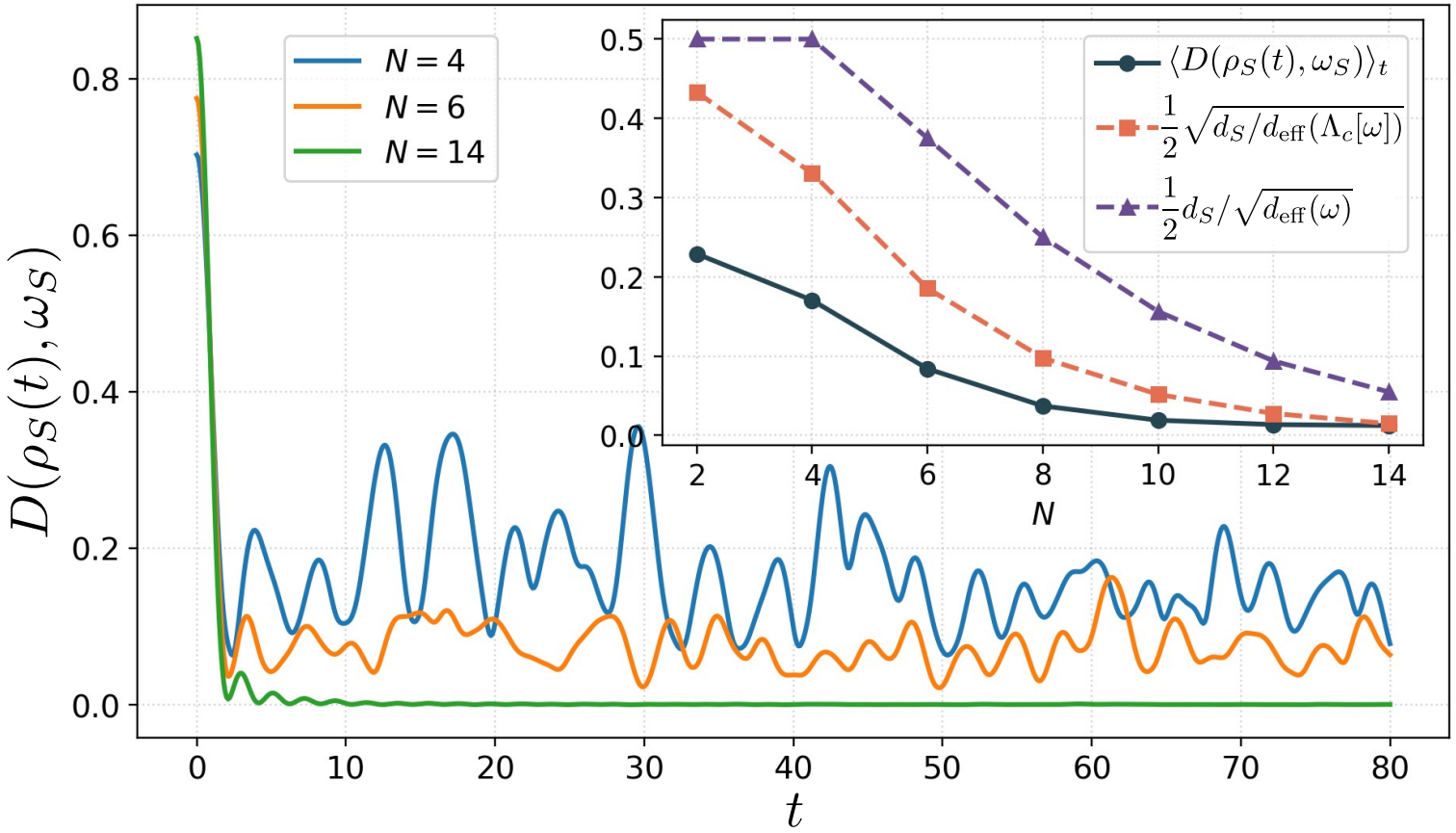}
    \end{minipage}
    \hfill
    \begin{minipage}{0.49\textwidth}
        \centering
        \includegraphics[width=\linewidth]{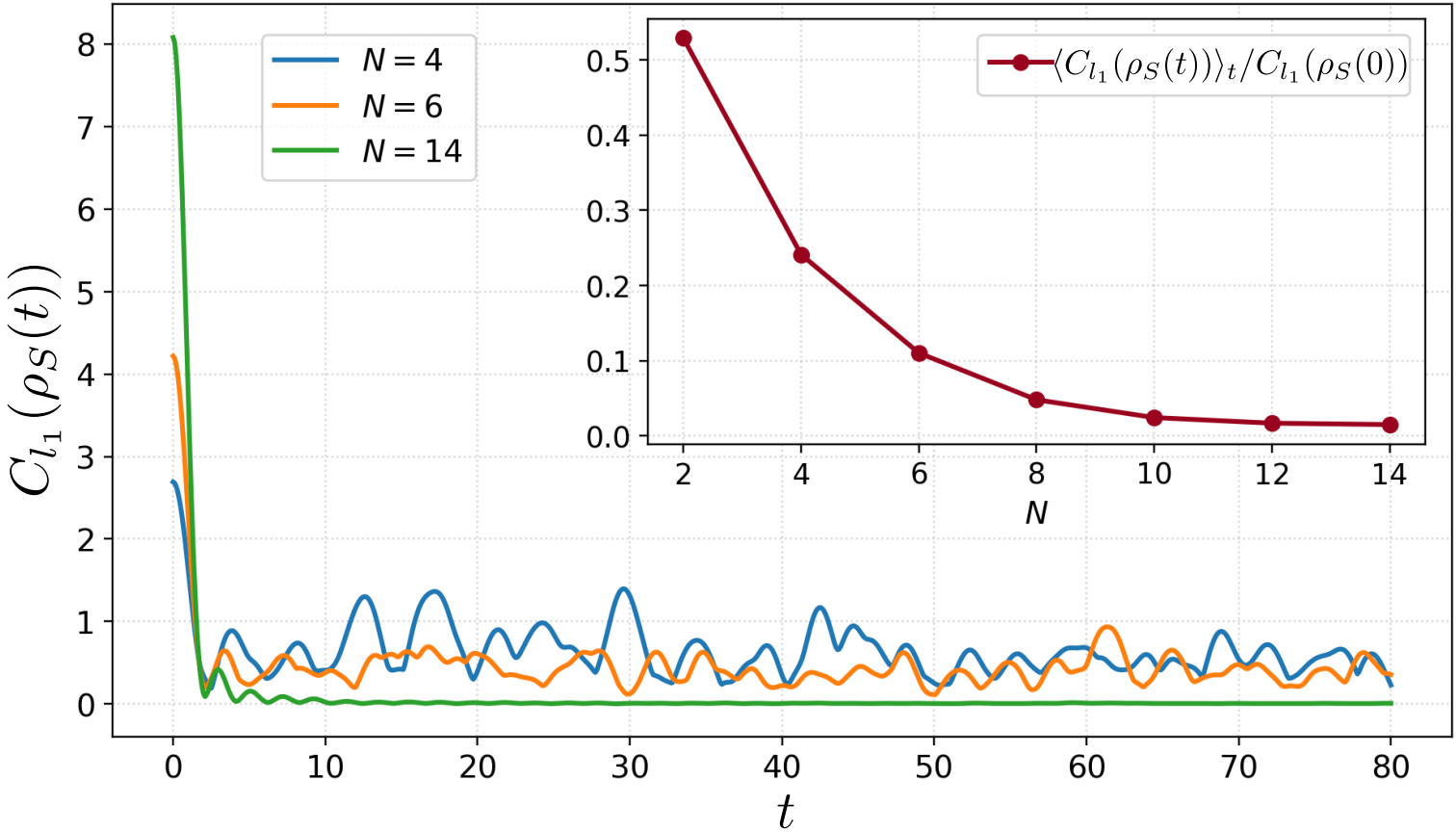}
    \end{minipage}
    \caption{Numerical illustration of the finite-resolution energy channel for a finite Ising spin chain. In both panels, the initial state is a balanced coherent superposition of all energy eigenstates, and the spectrum is partitioned into \(K=N\) resolved energy windows. Left: temporal behavior of the trace distance \(D(\varrho_S(t),\omega_S)\), together with its time average and the two equilibration bounds shown in the inset. Right: temporal behavior of the accessible coherence \(C_{l_1}(\varrho_S(t))\) between distinct effective energy levels, included as a diagnostic of the off-diagonal sector already controlled by trace-distance equilibration; the inset shows the normalized average coherence
\(\langle C_{l_1}(\varrho_S(t))\rangle_t/ C_{l_1}(\varrho_S(0))\).}
    \label{fig:energy-window-numerics}
\end{figure*}

The general equilibration result applies directly to the finite-resolution energy channel \(\Lambda_{\Delta E}\). Since \(d_S=K\) for this channel, Eq.~\eqref{eq:main-result-envfree} gives \(\langle D(\rho_S(t),\omega_S)\rangle_t\le K/(2\sqrt{d_{\mathrm{eff}}(\omega)})\). Thus, the finite-resolution description equilibrates whenever the number \(K\) of effective energy levels is small compared with \(\sqrt{d_{\mathrm{eff}}(\omega)}\), that is, the effective dimension of the system. This separation of scales is naturally expected in macroscopic many-body systems with local bounded interactions. Indeed, a finite energy resolution \(\Delta E\) gives \(K\sim (E_{\max}-E_{\min})/\Delta E\), and for local Hamiltonians with bounded interactions the total spectral width scales at most extensively, \(E_{\max}-E_{\min}=O(N)\)~\footnote{For a local Hamiltonian \(H=\sum_\ell h_\ell+\sum_\ell h_{\ell,\ell+1}\), with local terms of bounded norm, the operator norm of \(H\) is bounded by the sum of the norms of its local terms, which is \(O(N)\). Hence \(E_{\max}-E_{\min}\le 2\|H\|=O(N)\), which gives an extensive upper bound on the total spectral width.}. Therefore, for fixed finite resolution, the number of effective energy levels grows at most extensively with system size. When the initial state has significant weight over many microscopic energy eigenstates, \(d_{\mathrm{eff}}(\omega)\) becomes large, and the bound provides a strong equilibration estimate for the finite-resolution description.

The condition obtained from Eq.~\eqref{eq:main-result-envfree} is only sufficient, and it does not explicitly account for how the microscopic states are organized within each finite-resolution window. The refined bound in Eq.~\eqref{eq:main-result}, in contrast, retains information about this microscopic structure through the complementary output of the channel, which encodes the information that is discarded by the effective description, thereby avoiding the additional dimensional factor introduced in obtaining Eq.~\eqref{eq:main-result-envfree}.

This becomes especially clear in the population coarse-grained limit, \(\alpha_{ij}=0\), where the channel records only the resolved energy window \(j\) and discards the microscopic label \(\nu\) within that window. In this case, a Stinespring dilation can be chosen as \(V\ket{j,\nu}=\ket{j_\Lambda}\otimes\ket{\mathrm{aux}_{j\nu}}\). Thus, for \(\omega=\sum_{j,\nu}|c_{j\nu}|^2\ket{j,\nu}\bra{j,\nu}\), the complementary output is \(\Lambda_c[\omega]=\sum_{j,\nu}|c_{j\nu}|^2\ket{\mathrm{aux}_{j\nu}}\bra{\mathrm{aux}_{j\nu}}\), which has the same nonzero spectrum as \(\omega\). Therefore, \(d_{\mathrm{eff}}(\Lambda_c[\omega])=d_{\mathrm{eff}}(\omega)\). Importantly, this equality does not make the two bounds equivalent, because the same effective dimension enters Eqs.~\eqref{eq:main-result} and \eqref{eq:main-result-envfree} with different dependence on the generalized-subsystem dimension \(d_S\).

This distinction is seen directly by considering the finite-resolution description with \(K\) effective energy windows, each containing \(d\) unresolved microscopic levels, and an initial state typical in the corresponding \(Kd\)-dimensional subspace. Then, with high probability, \(d_{\mathrm{eff}}(\omega)\sim Kd/2\). Since \(d_S=K\), Eq.~\eqref{eq:main-result} gives the estimate on the left, while Eq.~\eqref{eq:main-result-envfree} gives the estimate on the right:
\begin{align}
    \frac{1}{2}\sqrt{\frac{K}{d_{\mathrm{eff}}(\Lambda_c[\omega])}}
    \sim
    \frac{1}{\sqrt{2d}},
    \qquad
    \frac{K}{2\sqrt{d_{\mathrm{eff}}(\omega)}}
    \sim
    \sqrt{\frac{K}{2d}} .
\end{align}
Thus, the refined bound only requires each effective energy window to contain many unresolved microscopic levels, \(d\gg1\), while the coarser bound requires the stronger condition \(d\gg K\). In particular, if \(K\) grows with \(d\), the refined bound may still show equilibration even when Eq.~\eqref{eq:main-result-envfree} is not informative.

We now turn to a numerical analysis of the finite-resolution energy channel. This analysis has two related aims: first, to illustrate how the equilibration behavior captured by Eqs.~\eqref{eq:main-result} and~\eqref{eq:main-result-envfree} emerges as the microscopic system grows. Second, to monitor the residual effective coherences resolved by \(\Lambda_{\Delta E}\). In particular, we will see that \(\langle D(\rho_S(t),\omega_S)\rangle_t\) becomes small and that the upper bounds become increasingly tight for relatively small finite systems. Meanwhile, the effective coherences that remain resolved after coarse-graining are strongly suppressed by the microscopic unitary dynamics.

We consider a system of \(N\) spin-\(1/2\) particles and choose as initial state a maximally coherent superposition of all energy eigenstates, \(\ket{\psi(0)}=d_U^{-1/2}\sum_{n=1}^{d_U}\ket{E_n}\), where \(d_U=2^N\). This choice gives an initial effective state sufficiently far from its time-averaged state, so that the decay of \(D(\rho_S(t),\omega_S)\) can be meaningfully followed. At the same time, because the coefficients have equal amplitudes and aligned phases, applying \(\Lambda_{\Delta E}\) produces appreciable off-diagonal terms between effective energy levels at \(t=0\). This prevents any later suppression of these terms from being a trivial consequence of an initially incoherent effective state.  At each time \(t\), the generalized-subsystem state \(\rho_S(t)\) is obtained by applying the finite-resolution channel \(\Lambda_{\Delta E}\) to the microscopic state \(\psi(t)\), using Eq.~\eqref{eq:energy-window-rhoS}.

For the microscopic dynamics, we use a one-dimensional Ising chain in transverse and longitudinal fields, with open boundary conditions,
\begin{align}
H
=
g\sum_{\ell=1}^{N}\sigma_x^{(\ell)}
+
h\sum_{\ell=1}^{N}\sigma_z^{(\ell)}
+
J\sum_{\ell=1}^{N-1}\sigma_z^{(\ell)}\sigma_z^{(\ell+1)} .
\end{align}
Here \(\sigma_x^{(\ell)}\) and \(\sigma_z^{(\ell)}\) are Pauli operators acting on site \(\ell\), and we set \(J=1\) and \(g=h=0.8\). With both transverse and longitudinal fields, this model provides a standard
finite-size spin-chain setting for studying statistical behavior in deterministic
quantum dynamics~\cite{jensen1985}. We use it here only as a
representative many-body model in which equilibration effects are readily
visible. The dynamics are obtained by exact diagonalization and implemented
using QuTiP~\cite{qutip}. The results are shown in Fig.~\ref{fig:energy-window-numerics} for finite chains with \(N=2,\ldots,14\) spins. For each \(N\), the spectrum is partitioned into \(K=N\) resolved energy windows, so that \(\Delta E=(E_{\max}-E_{\min})/N\).

The left panel illustrates the equilibration behavior. The main plot shows the time dependence of \(D(\rho_S(t),\omega_S)\) for representative system sizes. For the smallest chains, the trace distance still exhibits pronounced temporal fluctuations and no clear stationary regime. As \(N\) increases, these fluctuations are strongly reduced, and the effective state remains close to its time-averaged state for most times. The inset shows the corresponding time average \(\langle D(\rho_S(t),\omega_S)\rangle_t\) together with the two upper bounds in Eqs.~\eqref{eq:main-result} and~\eqref{eq:main-result-envfree}. It shows that the average trace distance decreases with increasing system size and that the bounds become increasingly tight for the finite systems considered here.

It is interesting to consider the timescale associated with the approach to the time-averaged state.
Seemingly, Fig.~\ref{fig:energy-window-numerics} suggests that the greater the initial distance $D(\rho_S(0),\omega_S)$, the faster its initial decay.
Crucially, however, this figure compares systems of different dimensions and thus the increase of $D(\rho_S(0),\omega_S)$ with $N$ may not imply a greater degree of nonequilibrium and therefore is not evidence of a Mpemba-like effect~\cite{Ares2023}.
Additionally, our numerical data indicates a saturation in the time of the first minimum of $D(\rho_S(t),\omega_S)$ as $N$ increases, suggesting that the initial approach to the time-averaged state occurs on an $N$-independent timescale for sufficiently large $N$.
We believe this behavior can be understood as follows: increasing $N$ initially enhances dephasing due to interference among a larger number of oscillatory contributions, but this efficiency eventually becomes limited by intrinsic timescales set by the Hamiltonian, leading to the observed saturation.
A precise determination of equilibration timescales in closed quantum systems is still an open problem~\cite{Short2012t,Reimann2016,Oliveira2018}, and since the time of the first minimum is not a rigorous measure of equilibration time, its saturation does not imply a finite equilibration time in the thermodynamic limit.
We leave a systematic investigation of this timescale and its dependence on the initial state for future work, including whether it exhibits the typicality observed in Markovian open systems~\cite{Bao2026}.

We now turn to the energetic coherences that remain accessible in the generalized subsystem. Motivated by the operational question of whether initially detectable macroscopic superpositions remain accessible under unitary dynamics~\cite{carvalho2025}, we ask an analogous question here for the effective coherences resolved by the finite-resolution generalized subsystem: do the residual coherences between effective energy levels persist for most times? To quantify the off-diagonal energy sector, we use the \(l_1\)-norm coherence measure~\cite{Baumgratz2014},
\begin{align}
    C_{l_1}(\rho_S)
    =
    \sum_{i\neq j}|(\rho_S)_{ij}|.
\end{align}
For the initial state chosen above, the equal amplitudes and aligned phases make the off-diagonal contributions add constructively after applying \(\Lambda_{\Delta E}\). Indeed, for \(\ket{\psi(0)}=d_U^{-1/2}\sum_{j,\mu}\ket{j,\mu}\), one has \((\rho_S(0))_{ij}=\alpha_{ij}d_i d_j/d_U\), for \(i\neq j\). Taking the maximal CPTP-compatible value \(|\alpha_{ij}|=1/\sqrt{d_i d_j}\), this gives \(C_{l_1}(\rho_S(0))=d_U^{-1}\sum_{i\neq j}\sqrt{d_i d_j}\). For approximately balanced windows, \(d_i\simeq d_U/K\), and therefore \(C_{l_1}(\rho_S(0))\simeq K-1\). Since in our example \(K=N\), the initially accessible effective coherence grows approximately linearly with \(N\), so the finite-resolution description contains an increasingly large coherent sector that is, in principle, accessible at \(t=0\).

The right panel shows how this initially accessible coherences behave under the subsequent dynamics. The main plot displays the time dependence of \(C_{l_1}(\rho_S(t))\) for representative system sizes. For smaller \(N\), a residual coherent contribution survives with temporal fluctuations. As \(N\) increases, however, the resolved coherence is rapidly suppressed and remains close to zero for most times. The inset shows the normalized time average \(\langle C_{l_1}(\rho_S(t))\rangle_t/C_{l_1}(\rho_S(0))\), that is, the fraction of the initially resolved effective coherence that survives along the dynamics. Its decrease with \(N\) shows that, even when the finite-resolution description initially retains a substantial amount of coherences, equilibration strongly suppresses its long-time contribution.

The numerical example therefore illustrates both aspects of the analysis: the finite-resolution generalized subsystem equilibrates in trace distance, while the same dynamics suppresses the residual coherent sector resolved by the channel. Thus, effective coherences that are initially resolved after coarse-graining need not remain operationally accessible for most times under the microscopic unitary dynamics.

\section{Conclusion}

Equilibration in isolated quantum systems is not a statement about the microscopic state becoming stationary. The microscopic state continues to evolve unitarily, but an observer usually has access only to a restricted effective description. In this work, we formulated this idea at the level of states by describing the accessible description as the output of a quantum channel. This defines a generalized subsystem and allows ordinary subsystems, finite families of POVMs, and coarse-grained descriptions to be treated within the same framework.

The resulting bounds show that equilibration is controlled by the dimension \(d_S\) of the generalized subsystem relative to the effective dimension explored by the microscopic dynamics: more sharply, through the complementary output \(\Lambda_c[\omega]\), and more generally, through the full time-averaged state \(\omega\). In the usual subsystem case, this information is carried by the environment. The sharper bound recovers the standard system--environment result, while the looser bound provides a sufficient criterion depending only on \(d_S\) and \(d_{\mathrm{eff}}(\omega)\). Thus, the familiar system--environment intuition survives beyond tensor-product subsystems: generalized subsystems equilibrate whenever their dimension is small compared with the relevant effective dimension generated by the microscopic dynamics.

The finite-resolution energy channel illustrates this mechanism in a physically motivated setting. A detector with finite spectral resolution does not resolve individual microscopic levels, but only energy windows. The unresolved multiplicities both restrict the coherences that can remain visible between effective energy levels and strengthen equilibration by increasing the effective dimension of the complementary output. Our numerical example shows this operationally: even coherences deliberately made visible at the effective level are suppressed under the unitary dynamics as the generalized subsystem equilibrates. In this sense, coherence need not disappear from the microscopic state in order to become inaccessible to the effective description.

\begin{acknowledgments}
This work is supported by the National Council for Scientific and Technological Development, CNPq Brazil (project: Universal Grant No. 408990/2025-2). A. D. Varizi acknowledges funding from the Air Force Office of Scientific Research under Grant No. FA9550-23-1-0092.
\end{acknowledgments}

\bibliographystyle{apsrev4-1}	
\bibliography{ref} 

\appendix

\section{Proofs of the equilibration bounds} 
\label{ap:mainresults}

\subsection{Bound in terms of the complementary output}

We first relate the trace distance to the Hilbert-Schmidt distance. For two states on a Hilbert space of dimension \(d_S\), we use the standard inequality~\cite{fuchs1999cryptographic}
\begin{align}
D(\rho_1,\rho_2)
=
\frac{1}{2}\|\rho_1-\rho_2\|_1
\leq
\frac{1}{2}
\sqrt{
d_S\,\tr\left[(\rho_1-\rho_2)^2\right]
}.
\label{norm1to2}
\end{align}
Applying this to \(\rho_1=\rho_S(t)\) and \(\rho_2=\omega_S\), taking the time average, and using Jensen's inequality for the concave function \(f(x)=\sqrt{x}\), we obtain
\begin{align}
\left\langle D(\rho_S(t),\omega_S)\right\rangle_t
&\leq
\frac{1}{2}\sqrt{d_S}\,
\left\langle
\sqrt{
\tr\left[(\rho_S(t)-\omega_S)^2\right]
}
\right\rangle_t
\nonumber\\
&\leq
\frac{1}{2}
\sqrt{
d_S
\left\langle
\tr\left[(\rho_S(t)-\omega_S)^2\right]
\right\rangle_t
}.
\label{eq:a2}
\end{align}

We now evaluate the time-averaged Hilbert-Schmidt term. From Eqs.~\eqref{eq:gsub} and~\eqref{eq:gsubeq},
\begin{align}
\rho_S(t)-\omega_S
=
\sum_{j\neq k}
c_jc_k^\ast e^{-i(E_j-E_k)t}
\Lambda\left(\ket{E_j}\bra{E_k}\right).
\label{eq:appendix-rhoS-minus-omegaS}
\end{align}
Therefore,
\begin{align}
&
\left\langle
\tr\left[(\rho_S(t)-\omega_S)^2\right]
\right\rangle_t
\nonumber\\
&=
\sum_{j\neq k}
\sum_{m\neq n}
c_jc_k^\ast c_m c_n^\ast
\left\langle
e^{-i(E_j-E_k+E_m-E_n)t}
\right\rangle_t
\nonumber\\
&\quad\times
\tr\left[
\Lambda\left(\ket{E_j}\bra{E_k}\right)
\Lambda\left(\ket{E_m}\bra{E_n}\right)
\right].
\label{eq:appendix-hs-expanded}
\end{align}
Using \(\left\langle f(t)\right\rangle_t
:=
\lim_{\tau\to\infty}
\frac{1}{\tau}
\int_0^\tau f(t)\,dt\), together with the non-degenerate-gap assumption, and since \(j\neq k\) and \(m\neq n\), the only nonzero contributions occur for \(j=n\) and \(k=m\). Hence
\begin{align}
&
\left\langle
\tr\left[(\rho_S(t)-\omega_S)^2\right]
\right\rangle_t
\nonumber\\
&=
\sum_{j\neq k}
|c_j|^2|c_k|^2
\tr\left[
\Lambda\left(\ket{E_j}\bra{E_k}\right)
\Lambda\left(\ket{E_k}\bra{E_j}\right)
\right].
\label{eq:appendix-hs-after-time-average}
\end{align}
Extending the sum to all \(j,k\) only adds nonnegative diagonal terms. Therefore,
\begin{align}
&
\left\langle
\tr\left[(\rho_S(t)-\omega_S)^2\right]
\right\rangle_t
\nonumber\\
&\leq
\sum_{j,k}
|c_j|^2|c_k|^2
\tr\left[
\Lambda\left(\ket{E_j}\bra{E_k}\right)
\Lambda\left(\ket{E_k}\bra{E_j}\right)
\right].
\label{eq:appendix-hs-extended-sum}
\end{align}

We now express the trace in Eq.~\eqref{eq:appendix-hs-extended-sum} in terms of the complementary channel. Let
\begin{align}
\Lambda[X]
=
\tr_{\rm aux}\left[VXV^\dagger\right],
\qquad
\Lambda_c[X]
=
\tr_S\left[VXV^\dagger\right]
\end{align}
be a Stinespring representation of \(\Lambda\) and its associated complementary channel. Defining \(\ket{\widetilde E_j}=V\ket{E_j}\), we have
\begin{align}
&
\tr\left[
\Lambda\left(\ket{E_j}\bra{E_k}\right)
\Lambda\left(\ket{E_k}\bra{E_j}\right)
\right]
\nonumber\\
&\qquad=
\tr\left[
\tr_{\rm aux}
\left(\ket{\widetilde E_j}\bra{\widetilde E_k}\right)
\tr_{\rm aux}
\left(\ket{\widetilde E_k}\bra{\widetilde E_j}\right)
\right].
\end{align}
For bipartite vectors \(\ket{a},\ket{b}\in\mathcal{H}_S\otimes\mathcal{H}_{\rm aux}\), one has
\begin{align}
&
\tr\left[
\tr_{\rm aux}\left(\ket{a}\bra{b}\right)
\tr_{\rm aux}\left(\ket{b}\bra{a}\right)
\right]
\nonumber\\
&\qquad=
\tr\left[
\tr_S\left(\ket{a}\bra{a}\right)
\tr_S\left(\ket{b}\bra{b}\right)
\right].
\label{eq:appendix-trace-identity}
\end{align}
Applying this identity to \(\ket{a}=\ket{\widetilde E_j}\) and \(\ket{b}=\ket{\widetilde E_k}\), we obtain
\begin{align}
&
\tr\left[
\Lambda\left(\ket{E_j}\bra{E_k}\right)
\Lambda\left(\ket{E_k}\bra{E_j}\right)
\right]
\nonumber\\
&\qquad=
\tr\left[
\Lambda_c\left(\ket{E_j}\bra{E_j}\right)
\Lambda_c\left(\ket{E_k}\bra{E_k}\right)
\right].
\label{eq:app-complementary-trace}
\end{align}
Substituting Eq.~\eqref{eq:app-complementary-trace} into Eq.~\eqref{eq:appendix-hs-extended-sum}, we find
\begin{align}
&
\left\langle
\tr\left[(\rho_S(t)-\omega_S)^2\right]
\right\rangle_t
\nonumber\\
&\leq
\sum_{j,k}
|c_j|^2|c_k|^2
\tr\left[
\Lambda_c\left(\ket{E_j}\bra{E_j}\right)
\Lambda_c\left(\ket{E_k}\bra{E_k}\right)
\right]
\nonumber\\
&=
\tr\left[
\Lambda_c\left(
\sum_j |c_j|^2\ket{E_j}\bra{E_j}
\right)
\Lambda_c\left(
\sum_k |c_k|^2\ket{E_k}\bra{E_k}
\right)
\right]
\nonumber\\
&=
\tr\left[\Lambda_c[\omega]^2\right]
=
\frac{1}{d_{\rm eff}(\Lambda_c[\omega])}.
\label{eq:appendix-hs-bound-complementary}
\end{align}
Finally, inserting Eq.~\eqref{eq:appendix-hs-bound-complementary} into Eq.~\eqref{eq:a2}, we obtain
\begin{align}
\left\langle D(\rho_S(t),\omega_S)\right\rangle_t
\leq
\frac{1}{2}
\sqrt{
\frac{d_S}{d_{\rm eff}(\Lambda_c[\omega])}
}.
\label{eq:appendix-main-bound}
\end{align}

\subsection{Bound in terms of the microscopic effective dimension}

We now derive a weaker bound depending only on the microscopic time-averaged state \(\omega\). The idea is to relate the purity of the complementary output to the purity of \(\omega\). We use the Rényi entropy
\[
S_\alpha(\rho)=(1-\alpha)^{-1}\ln\tr[\rho^\alpha],
\]
and the weak subadditivity relation: for a bipartite state \(\rho_{AB}\), and for all \(\alpha\)~\cite{van2002},
\begin{equation}
S_\alpha(\rho_A)-S_0(\rho_B)
\le
S_\alpha(\rho_{AB})
\le
S_\alpha(\rho_A)+S_0(\rho_B),
\label{eq:weak-subadditivity-renyi}
\end{equation}
where \(S_0(\rho)=\ln\mathrm{rank}(\rho)\).

Let \(V\) be the Stinespring isometry associated with \(\Lambda\), and define
\[
\omega_{\rm aux}=\tr_S[V\omega V^\dagger],
\qquad
\omega_S=\tr_{\rm aux}[V\omega V^\dagger].
\]
Thus \(\omega_{\rm aux}=\Lambda_c[\omega]\). Applying Eq.~\eqref{eq:weak-subadditivity-renyi} with \(\alpha=2\), \(A={\rm aux}\), and \(B=S\), we obtain
\begin{align}
\tr[(V\omega V^\dagger)^2]
&=
\exp[-S_2(V\omega V^\dagger)]
\nonumber\\
&\ge
\exp[-S_2(\omega_{\rm aux})-S_0(\omega_S)]
\nonumber\\
&=
\frac{\exp[-S_2(\omega_{\rm aux})]}
{\exp[S_0(\omega_S)]}
\nonumber\\
&=
\frac{\tr[\omega_{\rm aux}^2]}{\mathrm{rank}(\omega_S)}
\ge
\frac{\tr[\omega_{\rm aux}^2]}{d_S}.
\label{eq:purity-bound-renyi}
\end{align}
The first inequality uses
\(S_2(V\omega V^\dagger)\le S_2(\omega_{\rm aux})+S_0(\omega_S)\), while the last one uses \(\mathrm{rank}(\omega_S)\le d_S\).

Since \(V\) is an isometry, \(V^\dagger V=I\), and therefore
\begin{equation}
\tr[(V\omega V^\dagger)^2]
=
\tr[V\omega V^\dagger V\omega V^\dagger]
=
\tr[\omega^2].
\label{eq:isometry-purity}
\end{equation}
Combining Eqs.~\eqref{eq:purity-bound-renyi} and~\eqref{eq:isometry-purity}, we get
\begin{align}
\frac{1}{d_{\rm eff}(\omega_{\rm aux})}
=
\tr[\omega_{\rm aux}^2]
&\le
d_S\,\tr[(V\omega V^\dagger)^2]
\nonumber\\
&=
d_S\,\tr[\omega^2]
=
\frac{d_S}{d_{\rm eff}(\omega)}.
\label{eq:deff-aux-bound}
\end{align}
Equivalently,
\[
d_{\rm eff}(\omega_{\rm aux})
\ge
\frac{d_{\rm eff}(\omega)}{d_S}.
\]
Combining this with the previous bound, and using
\(\omega_{\rm aux}=\Lambda_c[\omega]\), we arrive at
\begin{equation}
\left\langle D(\rho_S(t),\omega_S)\right\rangle_t
\le
\frac{1}{2}
\sqrt{
\frac{d_S}{d_{\rm eff}(\Lambda_c[\omega])}
}
\le
\frac{1}{2}
\sqrt{
\frac{d_S^2}{d_{\rm eff}(\omega)}
}.
\label{eq:appendix-microscopic-effective-dimension-bound}
\end{equation}

\subsection{Typicality of Initial States:}

Here we prove Eq.~\eqref{eq:concentration-omegaS}.
As in the main text, let $\mathcal{H}_R\subseteq\mathcal{H}_U$ denote a subspace of dimension $d_R$ and $\overline{\bullet}^\psi$ denote the average over uniformly random pure states in this subspace.
For each $|\psi\rangle\in\mathcal{H}_R$, let $\omega^\psi$ denote the associated time-average state and $\omega_S^\psi\equiv \Lambda[\omega^\psi]$ denote the corresponding time-averaged state of the generalized subsystem.

To prove Eq.~\eqref{eq:concentration-omegaS} we need the following identities~\cite{Popescu2006,linden2009}
\begin{equation}\label{eq:id1}
    \tr\{AB\} = \tr\{(A\otimes B)\mathbb{S}\},
\end{equation}
where $\mathbb{S}|ij\rangle = |ji\rangle$ is the SWAP map, and~\cite{Popescu2006,linden2009}
\begin{equation}\label{eq:lemma}
    \overline{ \Pi_\psi \otimes \Pi_\psi }^\psi = \frac{\Pi_{RR}(\mathds{1}+\mathbb{S})}{d_R(d_R+1)},
\end{equation}
where $\Pi_\psi = |\psi\rangle \langle \psi|$, $\Pi_{RR}=\Pi_R\otimes\Pi_R$, $\Pi_R$ is the projector onto $\mathcal{H}_R$ and $\mathds{1}$ is the identity on $\mathcal{H}\otimes\mathcal{H}$.

Then, using again~\eqref{norm1to2} and Jensen's inequality we get
\begin{equation}\label{eq:previous1}
 \begin{aligned}
  \overline{D\!\left(\omega_S^\psi,\overline{\omega_S^\psi}^\psi\right)}^\psi &\le \frac{1}{2}\sqrt{d_S\,\overline{\tr_S\left\{\left(\omega_S^\psi - \overline{\omega_S^\psi}^\psi\right)^2\right\}}^\psi}
  \\[0.2cm]
  & = \frac{1}{2} \sqrt{d_S\, \tr_S\left\{\overline{\left(\omega_S^\psi\right)^2}^\psi - \left(\overline{\omega_S^\psi}^\psi\right)^2\right\}}.
 \end{aligned}
\end{equation}

Now,
\[
 \begin{aligned}
    & \tr_S \left\{ \! \overline{\left(\omega_S^\psi\right)^2}^\psi \!\!- \!\left(\overline{\omega_S^\psi}^\psi\right)^2 \!\right\}  =
    \\[0.3cm]
    & \tr_{SS}\! \left\{ \! \left[\overline{\omega_S^\psi\otimes\omega_S^\psi}^\psi - \overline{\omega_S^\psi}^\psi\otimes\overline{\omega_S^\psi}^\psi\right]\mathbb{S}\right\} =
    \\[0.3cm]
    & \tr_{SS}\! \left\{ \! \left[ \Lambda\circ \$ \otimes \Lambda \circ \$ \left[ \overline{|\Pi_\psi\otimes \Pi_\psi}^\psi - \overline{\Pi_\psi}^\psi \otimes \overline{\Pi_\psi}^\psi  \right] \right]\mathbb{S}\right\} =
    \\[0.3cm]
    & \tr_{SS}\! \left\{ \! \left[ \Lambda \circ \$ \otimes \Lambda \circ \$ \left[ \frac{\Pi_{RR} (\mathds{1} + \mathbb{S}) }{d_R(d_R+1)} - \frac{\Pi_{RR}}{d_R^2} \right] \right] \mathbb{S} \right\} \le
    \\[0.3cm]
    & \tr_{SS}\! \left\{ \! \left[ \Lambda \circ \$ \otimes \Lambda \circ \$ \left[ \frac{\Pi_{RR}\mathbb{S}}{d_R^2} \right]  \right] \mathbb{S} \right\} =
    \\[0.3cm]
    & \frac{1}{d_R^2} \tr_{SS}\! \left\{ \left[ \Lambda\otimes \Lambda \left[ \sum_{kl} |kl\rangle \langle kl| \Pi_{RR} |lk\rangle \langle kl| \right]  \right] \mathbb{S} \right\} =
    \\[0.3cm]
    & \frac{1}{d_R^2} \sum_{kl} \langle kl|\Pi_{RR} |lk\rangle \tr_{SS}\! \left\{ \left[ \Lambda[|k\rangle\langle k|] \otimes \Lambda[|l\rangle\langle l|] \mathbb{S} \right] \right\} =
    \\[0.3cm]
    & \frac{1}{d_R^2} \sum_{kl} \langle kl|\Pi_{RR} |lk\rangle \tr_S\! \left\{ \Lambda[|k\rangle\langle k|] \Lambda[|l\rangle \langle l|] \right\} \le
    \\[0.3cm]
    & \frac{1}{d_R^2} \sum_{kl} \langle kl|\Pi_{RR} |lk\rangle \tr_S\! \left\{ \frac{ \Lambda[|k\rangle\langle k|]^2 + \Lambda[|l\rangle \langle l|]^2 }{2} \right\} =
    \\[0.3cm]
    & \frac{1}{d_R^2} \sum_{k} \langle k|\Pi_R |k\rangle \tr_S\! \left\{ \Lambda[|k\rangle\langle k|]^2 \right\} \le
    \\[0.3cm]
    & \frac{1}{d_R^2} \sum_{k} \langle k|\Pi_R |k\rangle = \frac{1}{d_R}.
 \end{aligned}
\]
Throughout the derivation above we used the shorthand notation $|k\rangle\equiv|E_k\rangle$.
In the second and eighth line we used the identity~\eqref{eq:id1}.
In the third line we used that $\omega_S^\psi=\Lambda[\omega^\psi]$ and the fact that the infinite time-averaging in $\omega^\psi=\langle \Pi_\psi\rangle_t$ is equivalent to the action of the dephasing map $\$(\Pi_\psi) \equiv \sum_k |k\rangle \langle k| \Pi_\psi| k \rangle \langle k| = \omega^\psi$; \(\Lambda \circ \$[\bullet]\) denotes \(\Lambda[\$(\bullet)]\).
In the fourth line we used~\eqref{eq:lemma} and in the fifth that $d_R(d_R+1)>d_R^2$.
In the ninth line we used that $\tr\{AB\}\le (\tr\{A^2\} + \tr\{B^2\})/2$ and in the last one that the purity of a state is always no more than one.

Combining the result above with~\eqref{eq:previous1} we get Eq.~\eqref{eq:concentration-omegaS}.

\section{Complete positivity of the finite-resolution energy channel}
\label{ap:finite-resolution-cptp}

In this appendix we verify that the finite-resolution energy channel \(\Lambda_{\Delta E}\) is CPTP and derive the constraint imposed by complete positivity on the coefficients \(\alpha_{ij}\).

Trace preservation follows directly from the microscopic action of
\(\Lambda_{\Delta E}\). From
Eqs.~\eqref{eq:energy-window-map:diag} and
\eqref{eq:energy-window-map:off},
\begin{equation}
\tr\left[
\Lambda_{\Delta E}\left(\ket{i,\mu}\bra{j,\nu}\right)
\right]
=
\delta_{ij}\delta_{\mu\nu}
=
\tr\left[
\ket{i,\mu}\bra{j,\nu}
\right].
\end{equation}
By linearity,
\(\tr[\Lambda_{\Delta E}(X)]=\tr[X]\) for every operator \(X\).

The nontrivial part is to prove complete positivity. We do this through the Choi matrix, whose only nontrivial positivity condition reduces to the positivity of a finite \(K\times K\) matrix \(G\) defined by
\begin{equation}
G_{ii}=1,
\qquad
G_{ij}=\alpha_{ij}\sqrt{d_i d_j}
\quad (i\neq j).
\label{eq:G-main-result}
\end{equation}
Thus,
\begin{equation}
\Lambda_{\Delta E}\ \text{is CPTP}
\qquad\Longleftrightarrow\qquad
G\succeq0.
\label{eq:CPTP-iff-G}
\end{equation}
As an immediate pairwise consequence, complete positivity implies
\begin{equation}
|\alpha_{ij}|
\le
\frac{1}{\sqrt{d_i d_j}},
\qquad i\neq j,
\label{eq:alphaij-bound}
\end{equation}
which is the constraint quoted in the main text.

The proof proceeds by writing the Choi matrix in a form that separates two orthogonal sectors. One sector is manifestly positive, while the other contains all the coefficients \(\alpha_{ij}\). Positivity of this second sector is exactly the condition \(G\succeq0\).

\subsection{Choi matrix}

Let us now consider complete positivity. Let
\begin{equation}
J(\Lambda_{\Delta E})
=
\sum_{x,y}
\Lambda_{\Delta E}(\ket{x}\bra{y})
\otimes
\ket{x}\bra{y},
\end{equation}
with input basis \(\{\ket{x}\}\equiv\{\ket{i,\mu}\}\) and output basis \(\{\ket{i_\Lambda}\}\). Using the defining action of \(\Lambda_{\Delta E}\), we obtain
\begin{align}
J(\Lambda_{\Delta E})
&=
\sum_{i=0}^{K-1}
\ket{i_\Lambda}\bra{i_\Lambda}
\otimes
\sum_{\mu=1}^{d_i}
\ket{i,\mu}\bra{i,\mu}
\nonumber\\
&\quad+
\sum_{\substack{i,j=0\\ i\neq j}}^{K-1}
\alpha_{ij}
\ket{i_\Lambda}\bra{j_\Lambda}
\otimes
\sum_{\mu=1}^{d_i}
\sum_{\nu=1}^{d_j}
\ket{i,\mu}\bra{j,\nu}.
\label{eq:Choi:raw:DeltaE}
\end{align}
Every term has the matched form \(\ket{i_\Lambda}\bra{j_\Lambda}
\otimes
\ket{i,\mu}\bra{j,\nu}\). Hence \(J(\Lambda_{\Delta E})\) acts nontrivially only on
\begin{equation}
\mathcal R
=
\bigoplus_{i=0}^{K-1}
\left(
\ket{i_\Lambda}
\otimes
\mathrm{span}\{\ket{i,\mu}\}_{\mu=1}^{d_i}
\right)
\subset
\mathcal H_S\otimes\mathcal H_U .
\label{eq:R}
\end{equation}
On \(\mathcal R^\perp\), the Choi matrix acts as zero, so positivity can be checked on \(\mathcal R\).

\subsection{Orthogonal-sector decomposition of the Choi matrix}

For each window \(i\), define the normalized collective vector
\begin{equation}
\ket{\phi_i}
=
\frac{1}{\sqrt{d_i}}
\sum_{\mu=1}^{d_i}\ket{i,\mu}.
\end{equation}
The uniform microscopic sums in Eq.~\eqref{eq:Choi:raw:DeltaE} select this vector, since
\begin{equation}
\sum_{\mu=1}^{d_i}\ket{i,\mu}\bra{i,\mu}=I_{d_i},
\qquad
\sum_{\mu=1}^{d_i}
\sum_{\nu=1}^{d_j}
\ket{i,\mu}\bra{j,\nu}
=
\sqrt{d_i d_j}\,
\ket{\phi_i}\bra{\phi_j}.
\end{equation}
Therefore,
\begin{align}
J(\Lambda_{\Delta E})
&=
\sum_{i=0}^{K-1}
\ket{i_\Lambda}\bra{i_\Lambda}\otimes I_{d_i}
\nonumber\\
&\quad+
\sum_{\substack{i,j=0\\ i\neq j}}^{K-1}
\alpha_{ij}\sqrt{d_i d_j}\,
\ket{i_\Lambda}\bra{j_\Lambda}
\otimes
\ket{\phi_i}\bra{\phi_j}.
\label{eq:Choi:bright-before-splitting}
\end{align}

Now write
\(I_{d_i}
=
\left(I_{d_i}-\ket{\phi_i}\bra{\phi_i}\right)
+
\ket{\phi_i}\bra{\phi_i}
\). Defining \(\Pi_i^\perp=I_{d_i}-\ket{\phi_i}\bra{\phi_i}\) and introducing \(G\) as in Eq.~\eqref{eq:G-main-result}, the Choi matrix becomes
\begin{align}
J(\Lambda_{\Delta E})
&=
\sum_i
\ket{i_\Lambda}\bra{i_\Lambda}
\otimes
\Pi_i^\perp
+
\sum_{i,j}
G_{ij}\,
\ket{i_\Lambda}\bra{j_\Lambda}
\otimes
\ket{\phi_i}\bra{\phi_j}.
\label{eq:Choi:decomp:DeltaE}
\end{align}

This decomposition is orthogonal on \eqref{eq:R}. Indeed,
\[
\mathrm{span}\{\ket{i,\mu}\}_{\mu=1}^{d_i}
=
\mathrm{span}\{\ket{\phi_i}\}
\oplus
\mathrm{span}\{\ket{\phi_i}\}^{\perp}.
\]
The first term in Eq.~\eqref{eq:Choi:decomp:DeltaE} acts on the orthogonal
components and is manifestly positive, since it is a sum of projectors. The
second term acts on the collective subspace \(\mathrm{span}\{
\ket{i_\Lambda}\otimes\ket{\phi_i}
\}_{i=0}^{K-1},
\) and is the only part containing the inter-window coherence coefficients
\(\alpha_{ij}\). Hence all nontrivial complete-positivity constraints are
encoded in \(G\).

\subsection{Complete-positivity condition}

\begin{proposition}
\label{prop:CP-G-PSD}
The Choi matrix satisfies \(J(\Lambda_{\Delta E})\succeq0\) if and only if \(G\succeq0\).
\end{proposition}

\begin{proof}
Since \(J(\Lambda_{\Delta E})\) vanishes on \(\mathcal R^\perp\), it is enough to consider vectors in \(\mathcal R\). Let
\[
\ket{\psi}
=
\sum_{i=0}^{K-1}
\ket{i_\Lambda}\otimes\ket{x_i},
\qquad
\ket{x_i}\in\mathrm{span}\{\ket{i,\mu}\}_{\mu=1}^{d_i}.
\]
Decompose
\[
\ket{x_i}
=
\ket{x_i^\perp}
+
s_i\ket{\phi_i},
\qquad
\ket{x_i^\perp}\perp\ket{\phi_i},
\qquad
s_i=\braket{\phi_i}{x_i}.
\]
Using Eq.~\eqref{eq:Choi:decomp:DeltaE}, we find
\begin{align}
\bra{\psi}J(\Lambda_{\Delta E})\ket{\psi}
&=
\sum_i
\|\ket{x_i^\perp}\|^2
+
\sum_{i,j}
G_{ij}s_i^\ast s_j
\nonumber\\
&=
\sum_i
\|\ket{x_i^\perp}\|^2
+
s^\dagger Gs .
\label{eq:final-quad-form}
\end{align}
If \(G\succeq0\), then the right-hand side is nonnegative for all \(\ket{\psi}\), hence \(J(\Lambda_{\Delta E})\succeq0\). Conversely, if \(J(\Lambda_{\Delta E})\succeq0\), choosing \(\ket{x_i}=s_i\ket{\phi_i}\) gives
\[
\bra{\psi}J(\Lambda_{\Delta E})\ket{\psi}
=
s^\dagger Gs
\ge0
\]
for all \(s\in\mathbb C^K\). Therefore \(G\succeq0\).
\end{proof}

\subsection{Pairwise bound and useful choices}

The condition \(G\succeq0\) is the full complete-positivity constraint on the coefficients \(\alpha_{ij}\). A simple necessary consequence follows from each \(2\times2\) principal submatrix. For \(i\neq j\),
\[
G^{(ij)}
=
\begin{pmatrix}
1 & \alpha_{ij}\sqrt{d_i d_j}\\
\alpha_{ij}^\ast\sqrt{d_i d_j} & 1
\end{pmatrix}.
\]
Positivity of this block requires
\[
\det G^{(ij)}
=
1-|\alpha_{ij}|^2d_i d_j
\ge0,
\]
and therefore
\begin{equation}
|\alpha_{ij}|
\le
\frac{1}{\sqrt{d_i d_j}},
\qquad
i\neq j.
\label{eq:pairwise-alpha-bound}
\end{equation}
Higher-order principal minors may impose additional constraints on arbitrary choices of the \(\alpha_{ij}\), but Eq.~\eqref{eq:pairwise-alpha-bound} gives the sharp upper bound on each individual coefficient.

A useful compatible family is
\begin{equation}
\alpha_{ij}
=
\frac{x e^{i(\theta_i-\theta_j)}}{\sqrt{d_i d_j}},
\qquad
0\le x\le1,
\qquad
i\neq j.
\label{eq:alpha-x-family}
\end{equation}
For this family,
\[
G_{ii}=1,
\qquad
G_{ij}=x e^{i(\theta_i-\theta_j)}
\quad(i\neq j).
\]
Equivalently, with \(v_i=e^{i\theta_i}\),
\begin{equation}
G(x)
=
(1-x)I_K+xvv^\dagger.
\label{eq:Gx-decomposition}
\end{equation}
Thus, for any \(u\in\mathbb C^K\),
\begin{equation}
u^\dagger G(x)u
=
(1-x)\|u\|^2
+
x|v^\dagger u|^2
\ge0.
\end{equation}
Hence this whole family is compatible with complete positivity.

The limiting case \(x=0\) gives \(\alpha_{ij}=0\) and \(G=I_K\), corresponding to a population-only coarse-graining. The limiting case \(x=1\) gives the saturating coherent choice
\begin{equation}
\alpha_{ij}
=
\frac{e^{i(\theta_i-\theta_j)}}{\sqrt{d_i d_j}},
\qquad
i\neq j,
\label{eq:saturating_alpha}
\end{equation}
for which \(G=vv^\dagger\succeq0\). Thus both the population-only and saturating coherent cases are CPTP.

\end{document}